\documentclass[prodmode,acmec]{ec-acmsmallnofooter}
\usepackage[ruled]{algorithm2e}

\SetAlFnt{\small}
\SetAlCapFnt{\small}
\SetAlCapNameFnt{\small}
\SetAlCapHSkip{0pt}
\IncMargin{-\parindent}

\usepackage[numbers]{natbib}

\usepackage{algorithmic,bm,latexsym,amsmath,amssymb,graphicx,nicefrac,comment,mathtools,url,longtable}
\usepackage{array} 
\usepackage{supertabular} 

\newcommand{\Sh}{\mathit{Sh}}
\newcommand{\mc}{\mathit{mc}}
\newcommand{\CS}{\mathit{CS}}
\newcommand{\trim}{\mathit{trim}}

\newcommand{\Trimmed}{\mathit{Trimmed}}
\newcommand{\SameTrim}{\mathit{SameTrim}}
\newcommand{\ch}{\mathit{chi}}
\newcommand{\pa}{\mathit{par}}
\newcommand{\des}{\mathit{des}}
\newcommand{\anc}{\mathit{anc}}
\newcommand{\adj}{\mathit{adj}}
\newcommand{\Level}{\mathit{Level}}
\newcommand{\depth}{\mathit{depth}}
\newcommand{\height}{\mathit{height}}
\newcommand{\eqComment}[1]{\textnormal{\fontsize{7}{7}\selectfont{\ \ \ (#1)}}}

\newcommand{\victor}[1]{\textbf{[[Victor: #1]]}}

\begin{document}

\markboth{T. Rahwan et al.}{Towards a Fair Allocation of Rewards in Multi-Level Marketing}

\title{Towards a Fair Allocation of Rewards in Multi-Level Marketing}

\author{TALAL RAHWAN
\affil{Masdar Institute}
VICTOR NARODITSKIY
\affil{University of Southampton}
TOMASZ MICHALAK
\affil{University of Oxford}
MICHAEL WOOLDRIDGE
\affil{University of Oxford}
NICHOLAS R. JENNINGS
\affil{University of Southampton}
}

\begin{abstract} 
  An increasing number of businesses and organisations rely on
  existing users for finding new users or spreading a message.  One of
  the widely used ``refer-a-friend" mechanisms offers an equal reward
  to both the referrer and the invitee. This mechanism provides
  incentives for direct referrals and is fair to the invitee. On the
  other hand, multi-level marketing and recent social mobilisation
  experiments focus on mechanisms that incentivise both direct and
  indirect referrals. Such mechanisms share the reward for inviting a
  new member among the ancestors, usually in geometrically decreasing
  shares. A new member receives nothing at the time of joining. We
  study fairness in multi-level marketing mechanisms. We show how
  characteristic function games can be used to model referral
  marketing, show how the canonical fairness concept of the Shapley
  value can be applied to this setting, and establish the complexity of
  finding the Shapley value in each class, and provide a comparison of
  the Shapley value-based mechanism to existing referral mechanisms.
\end{abstract}
            

            




            

\maketitle


\section{Introduction}\label{sec:introduction}



Social networks and email made it easy to share information. These tools provide powerful ways to spread a message without using mainstream media. Compelling examples of their use can be found in politics (e.g., twitter revolutions), humanitarian relief (e.g., help-maps such as Ushahidi in post-earthquake Haiti), businesses (e.g., promoting products through Facebook), and entertainment (e.g., music videos such as Gangnam Style). Many famous cases of messages going viral are organic:  users share because they like the message, or identify with it.\footnote{For example, an absurd and funny Old Spice commercial collected nearly 50 million views \url{http://www.youtube.com/watch?v=owGykVbfgUE}
} However, it is hard to make a message viral by design, and in fact very few messages achieve viral stardom. A more realistic goal is a moderate level of social reach achieved be providing incentives for sharing information. 


The marketing literature acknowledges the importance of referral incentives~\cite{buttle98} but provides limited guidance on how to design them (see, e.g.,~\cite{biyalogorsky01}).  
Multi-level marketing mechanisms---mechanisms that offer not only direct but also indirect referrals---have recently received some attention in computer science.  Some of the interest has been fuelled by a  social mobilisation experiment known as the Red Balloon Challenge~\cite{darpa:net:report}, which provided validation for using multi-level mechanisms for referral incentivisation in a non-commercial context. The task posed by the challenge was to locate ten red weather balloons moored at undisclosed locations throughout the US. The winning incentive mechanism~\cite{science} shared the reward associated with each balloon in a recursive manner: the referrer received half of the finder's prize, the referrer's referrer received a quarter, etc. In the context of multi-level marketing this is an example of a  {\em geometric} reward mechanism. 

An axiomatic justification of geometric mechanisms appears
in~\cite{emek11}. An important property for such mechanisms is
\emph{Sybil-proofness}: i.e., the users should have no incentive to
create fake identities of themselves. Sybil-proof mechanisms have been
studied in a number of papers
including~\cite{babaioff12,drucker12,chen13,lv13}. Work on referral
incentives has also been carried out in the model of Query Incentive
Networks (QIN)~\cite{kleinberg05}, where the root of a tree needs to
incentivise the nodes to propagate the query until a node holding the
answer is reached. Performance of geometric mechanisms in QINs was
analyzed in~\cite{split}.

A number of desirable properties that a referral incentive mechanism
should satisfy are described in~\cite{douceur07}. One of the
properties relevant to fairness is ``value proportional to
contribution". This property ensures that a user recovers a fraction
of the effort he has put in. The effort, however, refers to performing
non-recruitment activities such as sharing files on a p2p network or
classifying galaxies on GalaxyZoo.

Unlike the above-cited work, our interest is in fairness when applied
to referrals as opposed to other types of contributions. In this work,
the contribution of a node is measured by the \emph{descendents} he
brings.  We formalise a class of cooperative games that enables us to
model such contributions. This class of games, which we call {\em tree
  games}\footnote{Not to be confused with {\em game trees}.} allows us
to formally study fairness in referral mechanisms by applying the
best-known game-theoretic fairness model---the Shapley
value~\cite{Shapley:53}.

The work is motivated not only by an academic interest in applying a
standard fairness concept to a new and popular domain, but also by the
fact that fairness of referral mechanisms is important in practice. In
particular, many popular mechanisms compensate not only the referrer
but also the invitee with an equal reward. Examples of such fair
``refer-a-friend" mechanisms include Dropbox and GiffGaff where both
the referrer and the invitee receive 500Mb of extra space or \pounds
5, respectively, for each successful referral of a new user.  As we
will show, compensation based on the Shapley value combines features
of both fair refer-a-friend mechanisms and geometric
mechanisms. Specifically, the Shapley value of an invitee is non-zero,
and all of the ancestors are compensated.

The following example illustrates our work.

\begin{example}
\label{ex:one}
The Dropbox referral program offers 500MB of extra free space to the
referrer and the invitee. Since a fixed amount is awarded for each
referral, we can say that Dropbox values each referral at 1GB per
user. Consider the referral tree on the right of
Figure~\ref{fig:graphAndInducedSubgraph}: $1$ invites $3$ who invites
$6$ and $7$.
\end{example}
There are 3 referrals, so Dropbox will distribute 3GB of free space.
Under the Dropbox's refer-a-friend scheme, user $1$ receives 500MB
when he invites $3$. User $3$ receives 500MB when accepting the
referral from $1$. Inviting $6$ and $7$ gives $3$ an extra 500MB more
for each invite. Users $6$ and $7$ receive 500MB at the time of
joining.  Distributing 3GB according to the geometric mechanism that
passes $\frac{1}{2}$ of referral reward to the referrer, results in
user $1$ receiving $\frac{1}{2}$ for $3$ and $\frac{1}{4}$ for each of
$6$ and $7$. User $3$ receives $\frac{1}{2}$ for each of $6$ and
$7$. The total shares for splitting 3GB are $(1,1,0,0)$ giving 1.5GB
to nodes $1$ and $3$ and nothing to the leaf nodes $6$ and $7$. The
resulting compensations are shown in Table~\ref{tbl:example}.

The current Dropbox mechanism is fair to the invitee: the invitee gets
the same reward as the referrer. However it is not fair to nodes who
bring many descendants, as a node is compensated only for direct
referrals but not for the nodes his children bring. On the other hand,
the geometric mechanism acknowledges indirect referrals (node 1 is
compensated for nodes 6 and 7) but does not provide any compensation
to the invitees at the time of joining. This can be viewed as unfair
by the invited node. Our interest is in defining a mechanism that is
fair to everyone. Every node joining through a referral brings value
to Dropbox, and each should acknowledged. However, nodes whose
referral activity brings more new customers should receive more. The
mechanism that we derive using the game-theoretic concept of Shapley
value satisfied both of these properties.

The derivation of the Shapley value for a general class of referral
marketing games follows in Section~\ref{sec:treeGames}, but for the
simple example given here, the Shapley-value mechanism has a simple
form: {\em the value of a referral is distributed in equal shares
  among the invitee and all of his ancestors}. Thus, when node $6$
joins, nodes $1$, $3$, and $6$ receive 333MB each. The final
distribution of the reward appears in Table~\ref{tbl:example}.


\begin{table}
    \centering
    {
        \tbl{Referrers' compensation for the tree on the right of Figure~\ref{fig:graphAndInducedSubgraph}.\label{tbl:example}}
        {
            \begin{tabular}{|l||c|c|c|c|}
                \hline
                & $1$ & $3$ & $6$ & $7$\\
                \hline
                Dropbox & $500$  &$1,500$ & $500$ & $500$ \\
                Geometric & $1,500$   & $1,500$ & $0$ & $0$\\
                Shapley & $1,167$ & $1,167$ & $333$ & $333$ \\
                \hline
            \end{tabular}
        }
    }
\end{table}

The example illustrates how referral marketing can be modelled by
referral trees. Specifically, there is one special node, the {\em
  root}, that initiates referral activities. This node can be viewed
as a customer who joined Dropbox independently. The root node can
invite its friends, and each successful referral has a value
associated with it that can be shared among the joining
nodes. However, if an invitation is not made, that is if an edge is
removed, then the entire bottom subtree does not join. This modelling
choice means that there is a single chance for each node to get
invited. This assumption is consistent with most models in the
referral literature
\cite{babaioff12,split,chen13,douceur07,drucker12,emek11,kleinberg05,lv13}. This
corresponds to modelling the referral process as trees where a node
can be referred by only one other node, as opposed to graphs allowing
multiple referrals. Considering graphs in general is an interesting
open question, but given that they occur naturally in many settings,
trees seem an obvious place to start an analysis.

A model related to ours was studied in \cite{myerson77}. \citeauthor{myerson77} defined games on graphs where the graph encodes which of the players can cooperate. The value of a coalition is then defined as the sum of the values of connected components in it. For the graph model, \citeauthor{myerson77} proved that there is one value that meets certain natural fairness axioms, and this value equals the Shapley value. We take the graph games to the referral marketing domain. Our model is a special case where the graph is a tree and the value of any connected component not containing the root is zero.  


Interestingly, since our model is a special case of \cite{myerson77}, the fairness
axioms defined there carry over to our marketing domain. In
particular, this means that any two players on both sides of the edge
profit from the edge equally. This property provides justification for
using the Shapley value in the referral marketing
context. ``Refer-a-friend" mechanisms such as Dropbox equally
compensate the referral and the invitee. The Shapley value on referral
trees generalises this fairness property to indirect referrals. Thus,
the Shapley value compensates \emph{all} types of contributions---it
is not only the direct referrals that matter but also the indirect
ones---and it provides a fair way of compensating the referrer and the
invitee in equal shares.

Our main contribution is a fair mechanism for referral
marketing. To this end, we introduce \textit{tree games} that allow us
to model rewards in multi-level marketing. We present a simple
technique for computing the Shapley value in the subclass where each
referral has a fixed value, and a simplified computation for the
general tree games. In the example above, each referral carried the same value of 1GB. Our
work extends to cases where the contribution of a referral depends on
the total number of nodes that join or, in the most general case on
the identities of all of the nodes that join.

The remainder of this paper is organised as
follows. Section~\ref{sec:background} is intended to familiarise the
reader with the necessary concepts from cooperative game theory. In
Section~\ref{sec:treeGames}, we introduce a class of games defined
over trees, and show how the Shapley value can be computed for these
games, thus providing a fair compensation for referral
activities. Section~\ref{sec:basicTreeGames} focuses on a natural
subclass of games where we show that the Shapley value is computable
in linear time. Section~\ref{sec:conclusions} concludes the paper and
presents some potential future
extensions. Appendix~\ref{sec:tableOfNotation} provides a summary of
the main notation used throughout the paper.


\section{Game Theoretic Background}\label{sec:background}


\noindent In this section, we will introduce some necessary definitions and concepts from cooperative game theory.

A \textbf{characteristic function game} is a pair $(A,v)$, where $A=\{1,\dots, n\}$ is a set of players (or \textit{agents}), and $v: 2^A\to{\mathbb R}$ is a \emph{characteristic function} that maps each subset (or \emph{coalition}) of agents, $C\subseteq A$, to a real number, $v(C)$, which represents the payoff of $C$. This number is called the \emph{value} of coalition $C$.


A \textbf{coalition structure}, $\CS\subseteq 2^A$, is a partition of $A$ into disjoint and exhaustive coalitions.
The set of possible coalition structures is denoted as $\mathcal{CS}^A$.

An \textbf{outcome} of a game is a pair, $(\CS,\mathbf{x})$, where $\CS\in\mathcal{CS}^A$ is a coalition structure, and $\mathbf{x}=(x_1, \dots, x_n)$ is a \textit{payoff vector}, which specifies a payoff, $x_i$, for every agent $i\in A$, such that: $x_i \geq 0$ for all $i=1,\dots,m$, and $\sum_{i\in C} x_i = v(C)$ for all $C\in \CS$. Intuitively, an outcome specifies which coalitions to form, and how the payoff of each coalition is divided among its members.

A \textbf{solution concept} specifies the set of outcomes that meet certain criteria. In this context, one desirable criterion is \textbf{stability}; an outcome, $(\CS,\mathbf{x})$, is said to be \textit{stable} if no group of agents can receive a payoff greater than what was allocated to them in that outcome, i.e., if $\sum_{i\in C}x_i \geq v(C), \forall C\subseteq A$. The set of all stable outcomes in a game is called \textbf{the core} of that game \cite{gillies}. In general, the problem of determining whether there exists a stable outcome is NP-complete \cite{Conitzer:Sandholm:03}.

\textbf{Fairness} is another desirable criterion when dealing with outcomes; it is typically evaluated based on the degree to which every agent's payoff reflects its contribution. In this vein, the best-known solution concept is the \textit{Shapley value}~\cite{Shapley:53}. To formally present this concept, we need to first introduce the notion of \textit{marginal contribution}.

\begin{definition}\label{def:marginalContribution}
\textbf{[Marginal Contribution]} Given a characteristic function game, $(A,v)$, the \emph{marginal contribution} of an agent $i\in A$ to a coalition $C\subseteq A\setminus\{i\}$ is denoted by $\mc(i,C)$, and defined as the difference in value that is caused when $i$ joins $C$. Formally:
\begin{equation}\label{eqn:marginalContribution}
\mc(i,C)=v(C\cup\{i\}) - v(C).
\end{equation}
\end{definition}
Finally, let $\Pi^A$ denote the set of all permutations of $A$. For any arbitrary permutation, $\pi\in\Pi^A$, let $\pi(i)$ be the location of $i$ in $\pi$, and let $C^{\pi}_i$ be the coalition consisting of the agents that precede $i$ in $\pi$. Formally,
$$
C^{\pi}_i=\{j\in A : \pi(j)<\pi(i)\}.
$$
\noindent Now, we are ready to define the Shapley value---a solution concept that specifies how to divide the value of the \textit{grand coalition}, i.e., the coalition containing all agents. According to the Shapley value, the payoff of an agent, $i\in A$, equals its average marginal contribution to the agents that precede it in an arbitrary permutation. This payoff is referred to as \textit{the Shapley value of $i$}. Formally:

\begin{definition}\label{def:ShapleyValue}
\textbf{[Shapley Value]}
Given a characteristic function game $(A,v)$, the \emph{Shapley value of an agent} $i\in A$ is denoted by $\Sh_i$ and is given by
\begin{equation}\label{eqn:shapleyValue1}
\Sh_i=\frac{1}{n!}\sum_{\pi\in\Pi^A}\mc(i,C^{\pi}_i).
\end{equation}
The outcome $(\{A\},(\Sh_1,\dots,\Sh_n))$ is called the Shapley value of $(A,v)$.
\end{definition} 

%

Due to its various desirable properties \cite{Shapley:53}, the Shapley value is arguably the fairest solution concept known to date in cooperative game theory .

%
%
%

A special class of coalitional games are \textbf{convex games} \cite{Shapley:71}. A characteristic function game $(A,v)$ is said to be convex if $v$ is \textit{supermodular}, i.e., if
$$
v(C\cup C') + v(C\cap C') \geq v(C)+v(C'),\ \ \forall C, C' \subseteq A.
$$
\noindent Some of the desirable properties of convex games include the fact that the Shapley value is always in the core, which implies that the core is always non-empty \cite{Shapley:71}.

\section{Tree Games}\label{sec:treeGames}


\noindent This section is divided into three subsections. The first introduces the main notation that will be used for graphs and trees. Section~\ref{sec:generalTreeGames} introduces a class of characteristic function games, which we call \textit{tree games},
%
%
while Section~\ref{sec:basicTreeGames} focuses on an intuitive subclass, namely \textit{basic tree games}.


\subsection{Main Notation for Graphs and Trees}\label{sec:graphNotation}


Throughout this paper, a graph will be denoted by $G$, with $N(G)$ and $E(G)$ being the set of nodes and the set of edges in $G$, respectively. In all the graphs that will be considered in this paper, each node will represent a unique agent. As such, a node will be denoted by a number $i\in\{1,\dots,n\}$ (just as we denote an agent), and a subset of nodes will be denoted by $C$ (just as we denote a coalition of agents). For every $C\subseteq N(G)$, we will write $\adj(C,G)$ to denote the nodes in $N(G)\setminus C$ that are each adjacent to one or more of the nodes in $C$. More formally,
\begin{equation}\label{eqn:adj}
\adj(C,G)=\{i\in N(G)\setminus C : \exists j\in C, (i,j)\in E(G)\}.
\end{equation}
Furthermore, we will write $G^C$ to denote the subgraph of $G$ that is induced by $C$. In other words, we have:
\begin{equation}\label{eqn:nodesOfSub}
N(G^C) = C,
\end{equation}
\begin{equation}\label{eqn:edgesOfSub}
E(G^C) = \{(i,j)\in E(G): i,j\in C\}.
\end{equation}
Next, we introduce our notation for trees. We use the notation $T$ instead of $G$ if the graph happens to be a rooted tree. Since every tree considered in this paper is rooted, we simply write ``tree'' instead of ``rooted tree'', and write $r(T)$ to denote the root of $T$. For every node $i\in N(T)$, let $\pa(i,T)$ be the parent of $i$ in $T$, and let $\ch(i,T)$, $\des(i,T)$ and $\anc(i,T)$ be the set of \textit{children}, \textit{descendants}, and \textit{ancestors} of $i$ in $T$, respectively.\footnote{\footnotesize Observe that $\anc(r(T),T)=\emptyset.$} 

Just as we denote by $G^C$ the subgraph of $G$ induced by $C$, we denote by $T^C$ the subgraph of $T$ induced by $C$. Furthermore, for every $i\in N(T)$, we write $T_i$ to denote the subset of $T$ rooted at $i$. More formally, $T_i = T^{\{i\}\cup\des(i,T)}$. We denote by $\depth(i,T)$ the depth\footnote{\footnotesize Recall that the depth of a node in a tree is the number of edges in the path between that node and the root of the tree. As such, the depth of the root itself is zero.} of $i$ in $T$, and by $\Level_j(T)$ the set of nodes in $N(T)$ whose depth is $j$. The height of $T$ is denoted by $\height(T)$. Formally, $\height(T)=\max_{i\in N(T)}\depth(i,T)$. 

Now, let us introduce the concept of \textit{``trimming''}. Basically, for every $C\subseteq N(T)$, by trimming $T^C$ we obtain another subgraph $T^C_{\trim}$, which is the connected component in $T^C$ that contains the root of $T$; if no such component exists then $T^C_{\trim}$ is the null graph. An example is illustrated in Figure~\ref{fig:graphAndInducedSubgraph}. More formally:

\begin{equation}\label{eqn:nodesOfSub_trim}
N(T^C_{\trim}) = \{i \in C: \anc(i,T)\subseteq C\},
\end{equation}
\begin{equation}\label{eqn:edgesOfSub_trim}
E(T^C_{\trim}) = \{(i,j)\in E(T): i,j\in N(T^C_{\trim}))\}.
\end{equation}	 

\begin{figure}[t!]
\centerline{\includegraphics{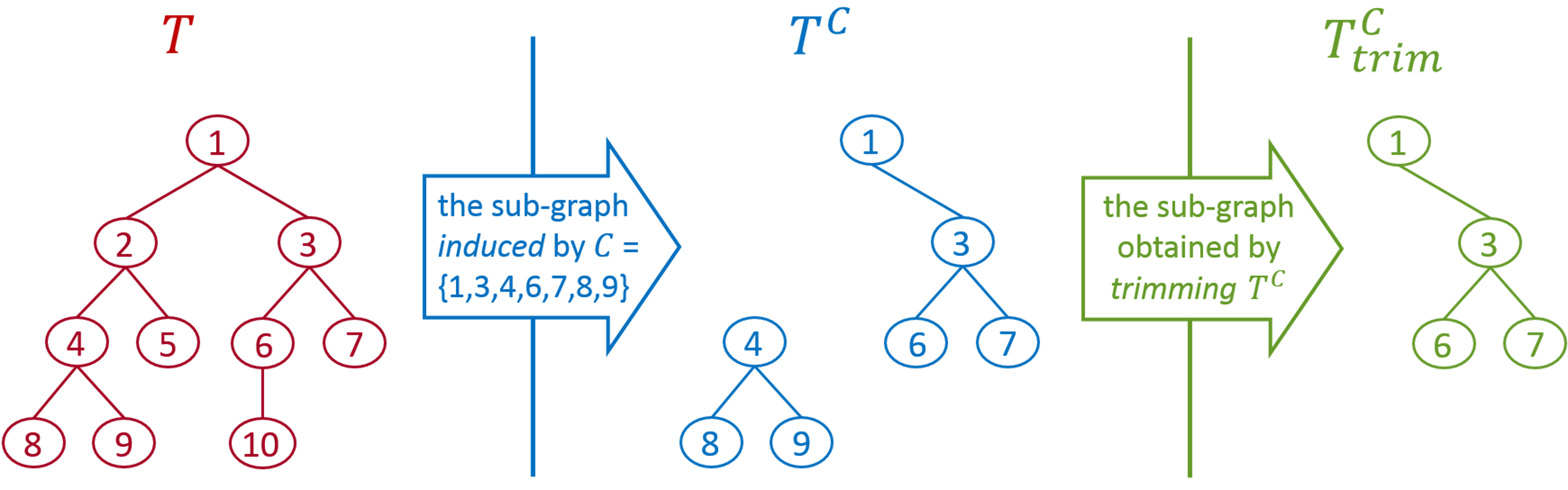}}
\caption{An example of $T$, $T^C$ and $T^C_{\trim}$, where $C=\{1,3,4,6,7,8,9\}$.}
\label{fig:graphAndInducedSubgraph}
\end{figure}


For every $C\subseteq N(T)$, the subgraph $T^C$ is said to be ``trimmed'' if and only if $T^C = T^C_{\trim}$. In Figure~\ref{fig:graphAndInducedSubgraph}, for example, the subgraphs $T^{\{1,3,6,7\}}$ and $T^{\{1,2,5\}}$ are trimmed, while the subgraph $T^{\{1,2,8,9\}}$ is not. We will denote by $\Trimmed(T)$ the set of coalitions that induce trimmed subgraphs. More formally:
\begin{equation}\label{eqn:trimmed}
\Trimmed(T) = \left\{C\subseteq N(T): T^C = T^C_{\trim}\right\}.
\end{equation}

Finally, for every $C\subseteq N(T)$, we will denote by $\SameTrim(C,T)$ the set consisting of every coalition $C'\subseteq N(C)$ such that by trimming $T^{C'}$ we obtain the same subgraph as the one obtained by trimming $T^C$. More formally:

\begin{equation}\label{eqn:sameTrim}
\SameTrim(C,T) = \{C'\subseteq N(T): T^{C'}_{\trim} = T^C_{\trim} \}.
\end{equation}

An example of $\Trimmed(T)$ and $\SameTrim(C)$ is illustrated in Figure~\ref{fig:exampleOfTrimmedAndSameTrim}. Appendix~\ref{sec:tableOfNotation} provides a summary all notation introduced thus far.


\begin{figure}[t!]
\centerline{\includegraphics{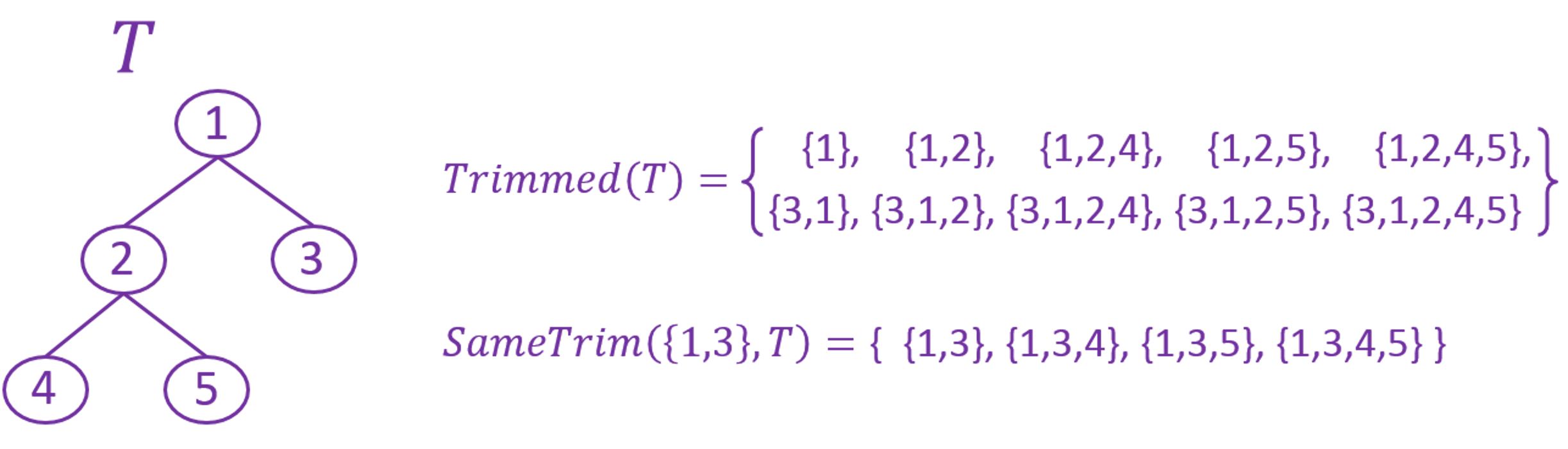}}
\caption{An example of $\Trimmed(T)$ and $\SameTrim(C)$, where $C=\{1,3\}$.}
\label{fig:exampleOfTrimmedAndSameTrim}
\end{figure}


\subsection{Tree Games}\label{sec:generalTreeGames}


We now show how trees induce a class of characteristic function games, which we call \textit{tree games}, and analyze some properties of these games.

\begin{definition}\label{def:treeGames}
\textbf{[Tree Game]} A \emph{tree game} is a pair $(T,f)$, where $T$ is a tree, and $f: \Trimmed(T)\to{\mathbb R}$. The set of agents in the tree game $(T,f)$ is $N(T)$, and the value of every coalition $C\subseteq N(T)$ is given by:
\begin{equation}\label{eqn:valueInTreeGames}
v(C) = f\left( N(T^C_{\trim}) \right).
\end{equation}
\end{definition}

Intuitively, a tree game $(T,f)$ is a characteristic function game, $(A,v)$, where $A=N(T)$ and the value of every $C\subseteq N(T)$ depends solely on the agents in $C$ that are connected to the root of $T$ via other members of $C$.


We will assume throughout this paper that $f(\emptyset) = 0$. As such, the following holds since $N(T^C_{\trim}) = \emptyset$ for all $C\subseteq N(T)\setminus\{r(T)\}$:
\begin{equation}\label{eqn:valueIsZeroWithoutTheRoot}
\forall C\subseteq N(T)\setminus\{r(T)\},\ v(C)=0.
\end{equation} 

The following lemmas will prove useful in subsequent proofs.


\begin{lemma}\label{lem:zeroContributionInTreeGames}
Let $(T,f)$ be a tree game. For every agent $i\in N(T)$, and every coalition $C\subseteq N(T):i\in C$, the following holds:
$$
i\notin N(T^C_{\trim})\ \ \ \ \Rightarrow\ \ \ \ \mc(i,C\setminus\{i\})=0.
$$
\end{lemma}


\begin{lemma}\label{lem:symmetricContributionInTreeGames}
Let $(T,f)$ be a tree game. For every agent $i\in N(T)$, and every pair of coalitions $C_1,C_2\subseteq N(T):i\in C_1\cap C_2$, the following holds:
$$
T^{C_1}_{\trim}=T^{C_2}_{\trim}\ \ \ \ \Rightarrow\ \ \ \ \mc(i,C_1\setminus\{i\})=\mc(i,C_2\setminus\{i\}).   
$$
\end{lemma}


\begin{lemma}\label{lem:SameTrimSetsAreDisjoint}
Let $(T,f)$ be a tree game. For every pair of coalitions $C_1, C_2\in \Trimmed(T):C_1\neq C_2$, the following holds:
$$
\SameTrim(C_1,T)\ \ \cap\ \ \SameTrim(C_2,T) \ \ = \ \ \emptyset.
$$
\end{lemma}


\begin{lemma}\label{lem:UnionOfSameTrimSetsCoversAllCoalitions}
Let $(T,f)$ be a tree game. The following holds, where $2^{N(T)}$ is the set of possible coalitions, i.e., subsets of $N(T)$:
$$
\bigcup\limits_{C\in\Trimmed(T)} \SameTrim(C,T) \ \ = \ \ 2^{N(T)}.
$$
\end{lemma}


\begin{lemma}\label{lem:1}
Let $(T,f)$ be a tree game. For every $i\in N(T)$, the following holds:
$$
\bigcup\limits_{C\in\Trimmed(T):i\in C} \SameTrim(C,T)\ \ \ =\ \ \ \{C\subseteq N(T) : i\in N(T^C_{\trim}) \}.
$$
\end{lemma}


\begin{lemma}\label{lem:2}
Let $(T,f)$ be a tree game, and let $C\in\Trimmed(T)$. For every coalition $C'\subseteq N(T)$, the following holds: 
$$
C'\in\SameTrim(C,T)\ \ \ \ \Leftrightarrow \ \ \ \ \big(C\subseteq C'\big)  \wedge  \big(C' \cap \adj(T^C_{\trim},T) = \emptyset\big)
$$
\end{lemma}


Having outlined the necessary lemmas, we now establish a key result with respect to the Shapley value in tree games.

\begin{theorem}\label{thm:shapleyOfTreeGames}
Let $(T,f)$ be a tree game. For every $i\in N(T)$, the following holds:
$$
\Sh_i = \sum\limits_{C \in \Trimmed(T):i\in C} \frac{\left|\adj(T^C,T)\right|!(\left|C\right|-1)!}{\left|C\cup\adj(T^C,T)\right|!} \Big( f(C) - f(C\setminus N(T_i)) \Big).
$$
\end{theorem}

\begin{proof}
%
%
According to Definition~\ref{def:ShapleyValue}, the Shapley value of $i$ can be computed using Equation~\eqref{eqn:shapleyValue1}. For convenience, we re-write this equation below (recall that $C^{\pi}_i$ is the coalition consisting of the agents that precede $i$ in the permutation $\pi$).
$$
\Sh_i=\frac{1}{n!}\sum_{\pi\in\Pi^A}\mc(i,C^{\pi}_i).
$$
However, for every pair of permutations $\pi_1,\pi_2\in\Pi^A$ such that $C^{\pi_1}_i = C^{\pi_2}_i$, we have: $\mc(i,C^{\pi_1}_i) = \mc(i,C^{\pi_2}_i)$. Based on this, the above equation can be written differently as follows:
$$
\Sh_i = \frac{1}{n!}\sum_{C\subseteq N(T):i\in C} \mc(i,C\setminus\{i\}) \times \left|\{\pi\in\Pi^A:C^{\pi}_i = C\setminus\{i\}\}\right|.
$$
Furthermore, for every $C\subseteq N(T):i\in C$, we know from Lemma~\ref{lem:zeroContributionInTreeGames} that $\mc(i,C\setminus\{i\})=0$ if $i\notin N(T^C_{\trim})$. The Shapley value of $i$ can then be computed as follows:
$$
\Sh_i = \frac{1}{n!}\sum_{C\subseteq N(T):i\in N(T^C_{\trim})} \mc(i,C\setminus\{i\}) \times \left|\{\pi\in\Pi^A:C^{\pi}_i = C\setminus\{i\}\}\right|.
$$
The above equation iterates over every $C\subseteq N(T):i\in N(T^C_{\trim})$. Instead, according to Lemma~\ref{lem:1}, it is possible to iterate over every $\{C\in\Trimmed(T):i\in C\}$, and for every such $C$, iterate over every coalition in $\SameTrim(C,T)$. In so doing, every coalition in $\{C\subseteq N(T):i\in N(T^C_{\trim})\}$ would be visited exactly once, according to Lemma~\ref{lem:SameTrimSetsAreDisjoint}. Based on this, we have:
{\fontsize{9}{9}\selectfont{
\begin{equation}\label{eqn:proof:shapleyOfTreeGames:1}
\Sh_i = \frac{1}{n!}\sum_{C\in\Trimmed(T):i\in C}\ \ \sum_{C'\in\SameTrim(C,T)} \mc(i,C'\setminus\{i\}) \times \Big|\{\pi\in\Pi^A:C^{\pi}_i = C'\setminus\{i\}\}\Big|.
\end{equation}
}}
Now, the definitions of $\Trimmed(T)$ and $\SameTrim(C,T)$, i.e., equations \eqref{eqn:trimmed} and \eqref{eqn:sameTrim}, imply that:
\begin{equation}\label{eqn:proof:shapleyOfTreeGames:2}
\forall C\in\Trimmed(T),\ C\in\SameTrim(C,T).
\end{equation}
Moreover, Lemma~\ref{lem:symmetricContributionInTreeGames} implies that: 
\begin{equation}\label{eqn:proof:shapleyOfTreeGames:3}
\forall C\in\Trimmed(T):i\in C, \forall C'\in\SameTrim(C,T),\ \mc(i,C'\setminus\{i\}) = \mc(i,C\setminus\{a_i\}).
\end{equation}

\noindent Based on equations \eqref{eqn:proof:shapleyOfTreeGames:2} and \eqref{eqn:proof:shapleyOfTreeGames:3}, we can write Equation~\eqref{eqn:proof:shapleyOfTreeGames:1} differently as follows:
$$
\Sh_i = \frac{1}{n!}\sum_{C\in\Trimmed(T):i\in C} \mc(i,C\setminus\{i\}) \sum_{C'\in\SameTrim(C,T)} \Big|\{\pi\in\Pi^A:C^{\pi}_i = C'\setminus\{i\}\}\Big|.
$$
This can be written differently as follows (because $i\in C$ and $i\notin C^{\pi}_i$):
{\fontsize{9}{9}\selectfont{
\begin{equation}\label{eqn:proof:shapleyOfTreeGames:5}
\Sh_i = \frac{1}{n!}\sum_{C\in\Trimmed(T):i\in C} \mc(i,C\setminus\{i\}) \times \Big| \{\pi\in\Pi^A : C^{\pi}_i \cup \{i\} \in \SameTrim(C,T)\} \Big|.
\end{equation}
}}
Next, based on Lemma~\ref{lem:2}, we will compute $\left|\{\pi\in\Pi^{N(T)}:C^{\pi}_i\cup\{i\}\in\SameTrim(C,T)\}\right|$ for a given coalition $C\in\Trimmed(T):i\in C$. To this end, for every $\pi\in\Pi^{N(T)}$, Lemma~\ref{lem:2} implies that $C^{\pi}_i\cup \{i\}\in\SameTrim(C,T)$ if and only if $C^{\pi}_i\cup\{i\}$ contains all the agents in $C$, and none of the agents in $\adj(T^C_{\trim},T)$. It is easy to see how the following steps cover all possible ways to construct such a permutation:
\begin{itemize}
\item\textbf{Step 1:} Place the members of $N(T)\setminus(C\cup\adj(T^C_{\trim},T))$ in any slots in $\pi$, without any restrictions.
\item\textbf{Step 2:} Out of all remaining slots, place the members of $\adj(T^C_{\trim},T)$ in the last slots, without any restriction on the order in which those members are placed in the slots. 
\item\textbf{Step 3:} Out of all remaining slots, place $i$ in the last slot. Steps 2 and 3 ensure that $C^{\pi}_i$ does not contain any of the agents in $\adj(T^C_{\trim},T)$.
\item\textbf{Step 4:} Place the members of $C\setminus\{i\}$ in the remaining slots in any order, without any restrictions. Steps 3 and 4 ensure that $C^{\pi}_i$ contains all the agents in $C$.
\item\textbf{Step 4:} Place the remaining members of $C$ (other than $i$) in the remaining slots in any order, without any restrictions. Steps 3 and 4 ensure that $C^{\pi}_i\cup\{i\}$ contains all the agents in $C$.

\end{itemize}
Let us count the number of possible ways in which each step can be performed:
\begin{itemize}
\item For step 1, the number is: $(n) (n-1) (n-2) \dots (n-\left|N(T)\setminus(C\cup\adj(T^C_{\trim},T))\right|+1)$.
\item For step 2, the number is: $\left|\adj(T^C_{\trim},T)\right|!$.
\item For step 3, the number is: 1.
\item For step 4, the number is: $(|C|-1)!$.
\end{itemize}
By multiplying those numbers, we obtain the size of the set $\{\pi\in\Pi^{N(T)}:C^{\pi}_i\cup\{i\}\in\SameTrim(C,T)\}$, where $C$ is a coalition in $\Trimmed(T)$ that contains $i$. Based on this, Equation~\eqref{eqn:proof:shapleyOfTreeGames:5} can be written differently as follows:
\begin{eqnarray}\label{eqn:proof:shapleyOfTreeGames:14}
     \Sh_i  &=&  \frac{1}{n!}\sum\limits_{{\scriptstyle C\in\Trimmed(T):i\in C}} \hspace*{-0.25cm}{\scriptstyle \mc(i,C\setminus\{i\}) \Big( \big(n\big)\dots\big(n-\left|N(T)\setminus(C\cup\adj(T^C_{\trim},T))\right|+1\big) \times \left|\adj(T^C_{\trim},T)\right|! \times \big(\left|C\right|-1\big)! \Big) } \nonumber\\
     &=& \sum\limits_{{\scriptstyle C\in\Trimmed(T):i\in C}}\ \mc(i,C\setminus\{i\})\ \ \frac{ \left|\adj(T^C_{\trim},T)\right|! \times \big(\left|C\right|-1\big)! }{ \big(n-\left| N(T)\setminus(C\cup\adj(T^C_{\trim},T))\right| \big)! } \nonumber\\
     &=& \sum\limits_{{\scriptstyle C\in\Trimmed(T):i\in C}}\ \mc(i,C\setminus\{i\})\ \ \frac{ \left|\adj(T^C_{\trim},T)\right|! \times \big(\left|C\right|-1\big)! }{ \left|C\cup\adj(T^C_{\trim},T)\right|! } \nonumber\\ 
     &=& \sum\limits_{{\scriptstyle C\in\Trimmed(T):i\in C}}\ \mc(i,C\setminus\{i\})\ \ \frac{ \left|\adj(T^C,T)\right|! \times \big(\left|C\right|-1\big)! }{ \left|C\cup\adj(T^C,T)\right|! } \eqComment{because $C\in\Trimmed(T)$}\nonumber
\end{eqnarray}
Based on this, in order to complete the proof of Theorem~\ref{thm:shapleyOfTreeGames}, it remains to show that: 
\begin{equation}\label{eqn:proof:shapleyOfTreeGames:16}
\mc(i,C\setminus\{i\}) = f\big( C \big) - f\big( C\setminus N(T_i) \big).
\end{equation}
To this end, based on the definition of marginal contribution and the definition of tree games, we know that for every tree game $(T,f)$ and every $C\subseteq N(T):i\in C$, the following holds:
\begin{equation}\label{eqn:proof:shapleyOfTreeGames:17}
\mc(i,C\setminus\{i\}) = f\big( N(T^C_{\trim}) \big) - f\big( N(T^{C\setminus\{i\}}_{\trim}) \big).
\end{equation}
We also know that $C\in\Trimmed(T)$. Thus, based on Equation~\eqref{eqn:trimmed}---the definition of $\Trimmed(T)$---as well as Equation~\eqref{eqn:proof:shapleyOfTreeGames:17}, we have:
$$
\mc(i,C\setminus\{i\}) = f\big(C\big) - f\big( N(T^{C\setminus\{i\}}_{\trim}) \big).
$$
Based on this, to prove the correctness of Equation~\eqref{eqn:proof:shapleyOfTreeGames:16}, it suffices to show that:
\begin{equation}\label{eqn:proof:shapleyOfTreeGames:18}
C\setminus N(T_i) \ =\ N(T^{C\setminus\{i\}}_{\trim}).
\end{equation}
From the definition of $\Trimmed(T)$, for every $C\in\Trimmed(T):i\in C$, we know that
\begin{equation}\label{eqn:proof:shapleyOfTreeGames:19}
C=N(T^C)=N(T^C_{\trim}),
\end{equation}
and it is easy to see that
\begin{equation}\label{eqn:proof:shapleyOfTreeGames:20}
N(T^{C}_{\trim})\setminus N(T_i) = N(T^{C\setminus\{i\}}_{\trim}).
\end{equation}
Equations \eqref{eqn:proof:shapleyOfTreeGames:19} and \eqref{eqn:proof:shapleyOfTreeGames:20} imply the correctness of Equation~\eqref{eqn:proof:shapleyOfTreeGames:18} and thus conclude the proof of Theorem~\ref{thm:shapleyOfTreeGames}.
\end{proof}

Now that we have defined tree games, and analyzed some of their relevant properties, in the following subsections we will focus on a subclass of tree games that is potentially relevant in some multi-level-marketing settings.


\subsection{Basic Tree Games}\label{sec:basicTreeGames}


%
Intuitively, \textit{basic tree games} are those where the value of a coalition, $C$, equals the number of agents in $C$ that are connected to the root via other members of $C$. More formally:

\begin{definition}\label{def:basic}
\textbf{[Basic Tree Game]} A tree game, $(T,f)$, is said to be \emph{basic} if and only if: $f(C)=|C|, \forall C\in\Trimmed(T)$.
\end{definition}

Our interest in this subclass is driven by our focus on multi-level marketing. To see how this is relevant, recall that one of main properties of the Shapley value is \emph{``Additivity''}---for every pair of games, $(N,v)$ and $(N,w)$, and every agent $i\in N$, we have: $\Sh_i(N,v) + \Sh_i(N,w) = \Sh_i(N,v+w)$, where the game $(N,v+w)$ is defined by $(v+w)(C) = v(C)+w(C)$ for every $C\subseteq N$. While this property is admittedly not very intuitive, it implies the linearity of the Shapley value. That is, if we scale a game---i.e., multiply all coalition values by some constant---the Shapley values will simply be multiplied by that same constant. To see how this relates to basic tree games, consider a website like \url{www.888casino.com}, which claims to offer the same amount of reward, $\pounds 88$ to be precise, for every subscriber. This can be interpreted as follows. For every subset of agents $C\subseteq N(T)$, the reward is $\pounds 88$ multiplied by the number of agents in $C$ who reach the root via other members of $C$ (other members of $C$ receive no reward simply because they did not reach the root). In other word, in the tree game $(T,f)$ which represents this website's preferences, we have:
$$
f(C) = \pounds 88 \times|C|,\ \ \forall C\in\Trimmed(T).
$$
Now instead of computing the Shapley values in this game, it is possible---based on the aforementioned linearity property---to first compute the Shapley values in an alternative basic tree game $(N,f')$, where:
$$
f'(C) = \pounds 1 \times |C|,\ \ \forall C\in\Trimmed(T),
$$
and then simply multiply the resulting Shapley values by $88$. This would return the Shapley values of the original game, $(T,f)$. More generally, for every tree game $(T,f)$ where there is a constant reward $x$ for every player that reaches the root, it is possible to first compute the Shapley values in a basic tree game which has the same tree, $T$, and then multiply the resulting Shapley values by $x$.

Now that we have elaborated on the intuition behind tree games, let us analyze some of their properties.

\begin{lemma}\label{lemma:marginalContribution:basicGame}
In a basic tree game, $(T,f)$, for every agent $i\in N(T)$, and every coalition $C\subseteq N(T)\setminus\{i\}$, the following holds:
$$
\mc(i,C)= \Big|N(T^{C\cup\{i\}}_{\trim})\cap N(T_i)\Big|.
$$
\end{lemma}

\begin{proof}
Following the definitions of marginal contribution and basic tree games:
\begin{equation}\label{eqn:proof:marginalContribution:basicGame:1}
\mc(i,C) = \Big|N(T^{C\cup\{i\}}_{\trim})\Big| - \Big|N(T^C_{\trim})\Big|.
\end{equation}
Furthermore, we know from Equation~\eqref{eqn:nodesOfSub_trim}---the definition of $N(T^C_{\trim})$---that the following holds, simply because $C$ is a subset of $C\cup\{i\}$:
\begin{equation}\label{eqn:proof:marginalContribution:basicGame:2}
N(T^C_{\trim})\subseteq N(T^{C\cup\{i\}}_{\trim}).
\end{equation}
Based on \eqref{eqn:proof:marginalContribution:basicGame:1} and \eqref{eqn:proof:marginalContribution:basicGame:2}, we have: 
\begin{equation}\label{eqn:proof:marginalContribution:basicGame:3}
\mc(i,C) = \Big|N(T^{C\cup\{i\}}_{\trim}) \setminus N(T^C_{\trim})\Big|.
\end{equation}
We also know from Equation~\eqref{eqn:nodesOfSub_trim} that an agent $j\in N(T^{C\cup\{i\}}_{\trim})$ does not belong to $N(T^C_{\trim})$ if and only if $i=j$ or $i\in\anc(j,T)$; in either case we have $j\in N(T_i)$. This, as well as Equation~\eqref{eqn:proof:marginalContribution:basicGame:3}, imply the correctness of the lemma.
\end{proof}

One of the desirable properties of a basic tree game is that it is convex, as stated in the following theorem.

\begin{theorem}\label{thm:basicGamesAreConvex}
A basic tree game is a convex game.
\end{theorem}

\begin{proof}
It is known (see, e.g., \cite{elkind:etal:13}) that a characteristic function game is convex if and only if for every pair of coalitions $C$, $C'$ such that $C\subset C'$, and for every agent $i$ not belonging to $C'$, the following holds:
\begin{equation}\label{eqn:proof:basicGamesAreConvex}
\mc(i,C) \leq \mc(i,C').
\end{equation}
\noindent Now since $C\subset C'$, then based on Equation~\eqref{eqn:nodesOfSub_trim}, we have:
$$
N(T^{C\cup\{i\}}_{\trim})\subseteq N(T^{C'\cup\{i\}}_{\trim}).
$$
\noindent This, as well as Lemma~\ref{lemma:marginalContribution:basicGame}, imply that the inequality in \eqref{eqn:proof:basicGamesAreConvex} holds, which implies the correctness of Theorem~\ref{thm:basicGamesAreConvex}.
\end{proof}

Theorem~\ref{thm:basicGamesAreConvex} immediately implies that the Shapley value of a basic tree game is in the core of that game. This makes the Shapley value-based division scheme even more attractive, since it is not only fair but also stable (see Section~\ref{sec:background} for more details).

\begin{theorem}\label{thm:shapleyOfBasicGames}
Let $(T,f)$ be a basic tree game. For every agent $i\in N(T)$, we have:
\begin{equation}\label{eqn:shapley:basicGame}
\Sh_i = \sum_{j=0}^{\height(T_i)} \frac{\left|\Level_j(T_i)\right|}{\depth(i,T)+j+1}
\end{equation}
\end{theorem}

\begin{proof}
Since $(T,f)$ is a basic tree game, Equation~\eqref{eqn:shapleyValue1} and Lemma~\ref{lemma:marginalContribution:basicGame} imply that:
$$
\Sh_i = \frac{1}{n!}\sum_{\pi\in\Pi^{N(T)}} \Big|N(T^{C^{\pi}_i\cup\{i\}}_{\trim})\cap N(T_i)\Big|.
$$
\noindent This can be written differently as follows:
$$
\Sh_i = \frac{1}{n!} \sum\limits_{j\in N(T_i)} \Big|\{\pi\in\Pi^{N(T)} : j\in N(T^{C^{\pi}_i\cup\{i\}}_{\trim})\}\Big|.
$$
Based on this, as well as \eqref{eqn:nodesOfSub_trim}---the equation that defines $N(T^C_{\trim})$---we have:
\begin{equation}\label{eqn:proof:shapleyForBasicGames:2}
\Sh_i = \sum\limits_{j\in N(T_i)} \frac{\left|\Pi^{ij}\right|}{n!} \ ,
\end{equation}
\noindent where
\begin{equation}\label{eqn:proof:shapleyForBasicGames:3}
\Pi^{ij} = \left\{ \pi\in\Pi^{N(T)} : \anc(j,T) \subseteq C^{\pi}_i\cup\{i\} \right\}.
\end{equation}
%
%
Let us compute $\left|\Pi^{ij}\right|/n!$. Based on Equation~\eqref{eqn:proof:shapleyForBasicGames:3}, for every $\pi \in \Pi^{N(T)}$, we have:
$$
\pi \in \Pi^{ij}\ \ \Leftrightarrow\ \ \anc(j,T)\setminus\{i\} \subseteq C^{\pi}_i.
$$
\noindent Thus, the number of permutations in $\Pi^{ij}$ equals the number of permutations in which all the agents in $\anc(j,T)\setminus\{i\}$ appear before $i$ in $\pi$. It is easy to see how the following steps cover all possible ways to construct such a permutation:
\begin{itemize}
\item\textbf{Step 1:} Place the members of $N(T)\setminus(\anc(j,T)\cup\{i\})$ in any slots in $\pi$, without any restrictions.
\item\textbf{Step 2:} Out of all remaining slots, place $i$ in the last one. 
\item\textbf{Step 3:} Place the members of $\anc(j,T)$ in the remaining slots (from the previous step, all these slots are before $i$) in any order, without any restrictions. 
\end{itemize}
Based on this, we have:
$$
\left|\Pi^{ij}\right| = \big((n)(n-1)\dots(n-|N(T)\setminus(\anc(j,T)\cup\{i\})|+1)\big) \cdot \big(\left|\anc(j,T)\right|!\big).
$$
\noindent This, as well as the fact that $|\anc(j,T)|=\depth(j,T)$, imply that:

$$
\left|\Pi^{ij}\right| = \big(n\big)\dots\big(\depth(j,T)+2\big)\big(\depth(j,T)\big)\dots\big(1\big).
$$

\noindent Based on this, we have:
$$
\frac{\left|\Pi^{ij}\right|}{n!} = \frac{(n)\dots(\depth(j,T)+2)(\depth(j,T))\dots(1)}{(n)\dots(1)}
$$
\noindent This can be simplified as follows:
\begin{equation}\label{eqn:proof:shapleyForBasicGames:8}
\frac{\left|\Pi^{ij}\right|}{n!} = \frac{1}{depth(j,T)+1}
\end{equation}
Finally, since $j\in N(T_i)$, we have:
\begin{equation}\label{eqn:proof:shapleyForBasicGames:9}
\depth(j,T)=\depth(i,T)+\depth(j,T_i).
\end{equation}
From \eqref{eqn:proof:shapleyForBasicGames:2}, \eqref{eqn:proof:shapleyForBasicGames:8} and \eqref{eqn:proof:shapleyForBasicGames:9}, we find that:
$$
\begin{array}{lll}
\Sh_i & = & \sum\limits_{j\in N(T_i)} \frac{1}{depth(j,T)+1}\smallskip\\
       & = & \sum\limits_{j\in N(T_i)} \frac{1}{\depth(i,T)+\depth(j,T_i)+1}\smallskip\\
       & = & \sum\limits_{j=0}^{\height(T_i)}\sum\limits_{j\in Level_j(T_i)} \frac{1}{\depth(i,T)+j+1}\smallskip\\
       & = & \sum\limits_{j=0}^{\height(T_i)} \frac{\left|Level_j(T_i)\right|}{\depth(i,T)+j+1}
\end{array}
$$
\end{proof}

Based on Theorem~\ref{thm:shapleyOfBasicGames}, as well as the fact that $\height(T_i)\leq n$ for all $i\in N(T)$, the Shapley value of any agent in a basic tree game is computable in $O(n)$ time.

\subsection{Computational Complexity of Shapley Value}
\noindent We have shown that the Shapley value for basic tree games is
linear-time computable. In this section, we compare the computational
complexity of the Shapley in characteristic function games (CFG), tree
games and basic tree games. We consider three standard tree
structures: chain, star, and compete binary
tree. Table~\ref{tbl:complexity} shows the number of coalitions that
need to be considered when computing the Shapley value of agent $i$ at
depth $d$ of a tree with height $h=\height(T)$ when the total number
of agents is $n$. The computation for the case of a complete binary
tree in tree games appears below:
\begin{align*}
& b(h,d)  = (y_{h-d}+1)^2\times \prod_{j=h-d+1}^h(y_j+1), \label{eq:b}
\end{align*}
where $y_1 = 1$ and  $y_j = (y_{j-1}+1)^2$ for $j>1$.

\begin{table}
    \centering
    {
        \tbl{Complexity of the Shapley Value.\label{tbl:complexity}}
        {
            \begin{tabular}{|l||c|c|c|}
                \hline
                & CFG& Tree Game & Basic Tree Game\\
                \hline
                Chain & $2^{n-1}$  &$h-d+1$ & $h-d+1$ \\
                Star &$2^{n-1}$   & $2^{n-2}$& $2-d$\\
                Binary Tree & $2^{n-1}$ & $b(h,d)$ & $h-d+1$ \\
                Arbitrary tree & $2^{n-1}$   &  $\left| \{C \in \Trimmed(T) : i\in C\} \right|$ & $\height(T_i)+1$  \\
                \hline
            \end{tabular}
        }
    }
\end{table}

\section{Discussion and Conclusions}\label{sec:conclusions}



\noindent
We now revisit Example~\ref{ex:one}. The fact that Dropbox values a coalition of extra users at 1GB per user, can be encoded as $f(C) = |C|$ satisfying the definition of basic tree games (see Definition~\ref{def:basic}). Referring to Table~\ref{tbl:example}, the Dropbox mechanism give the highest reward to the user who made the most direct referrals: $3$ gets three times as much as any of the other user. Users $6$ and $7$ who did not invite anyone receive 500MB each. In contrast, the geometric mechanism rewards both direct an indirect referrals but does not give anything to non-referring nodes. The Shapley value provides a middle ground: it acknowledges indirect referrals and also compensates non-referring nodes. The Shapley rewards are fair in the sense that  extra value is assigned to $1$ relative to $3$ because without the referral of $1$, agent $3$ would have never joined. The values for the Shapley rewards\footnote{For the Dropbox example, where joining independently does not earn the 1GB referral bonus, we need to modify \eqref{eqn:shapley:basicGame} by subtracting 1 (or 1GB) from the value of the root. So the values of non-root nodes are given by \eqref{eqn:shapley:basicGame} and the value of the root node is given by \eqref{eqn:shapley:basicGame} minus 1.} appearing in Table~\ref{tbl:example} can be computed using \eqref{eqn:shapley:basicGame}. The rewards can also be computed dynamically as each new user joins. Should node $7$ bring a friend, the compensation of each node would change by adding 250MB to compensations of the new node and of nodes $1$, $3$,  $7$.

We derived a mechanism based on a fairness concept. However, it also provides participation incentives for the agents. The property of the Shapley value that in compensates the ``right" types of contribution (compensation is based on the entire subtree and not just on direct referrals) makes it more appealing for successful incentivisation from the intuitive standpoint. Observe that without introducing some costs into the model, any mechanism that offers a non-zero reward to the referrer and the invitee provides rational agents with incentives to participate and to refer all of their friends. To differentiate among Shapley and other mechanisms a future study may consider models where costs (e.g., of making a referral) are present.


General tree games allow applying the Shapley value to scenarios where the firm's value for a new user depends on the identity of the user and the identities of other users who join (this allows expressing  combinatorial preferences). This is relevant, for example, when the firm has a particularly high value for certain celebrities who shape the image of the firm. This also allows modelling domains like the Red Balloon Challenge and QINs, where the value is obtained only if a node that holds the answer joins.

The Shapley value has a  strong axiomatic justification, but its use in practice is limited. The referral marketing domain may be a compelling setting for applying the Shapley value. Empirical evaluation of its merits is an interesting direction for future investigation. 

Implicit in our results is the assumption that each user has only one chance to receive a referral, or in other words, that a node has a single parent through which he can be referred. Relaxation of this assumption is a possible direction for further theoretical work.


\bibliographystyle{acmsmall}
\bibliography{bibliography_Shapley}


\newpage
\renewcommand{\thesection}{Appendix \arabic{section}}
\appendix
\section{Summary of Notation}\label{sec:tableOfNotation}
\begin{longtable}{@{\hspace{0cm}}>{\centering\arraybackslash}m{2.5cm}@{\hspace{0.5cm}}>{\raggedright\arraybackslash}m{10.8cm}}
$A$ & the set of agents in the game.\smallskip\\
$n$ & the number of agents in the game.\smallskip\\
$C$ & a coalition, i.e., a subset of $A$.\smallskip\\ 
$v(C)$ & the value of coalition $C$.\smallskip\\
$\CS$ & a coalition structure, i.e., a partition of $A$.\smallskip\\
$\mathcal{CS}^A$ & the set of coalition structures over $A$.\smallskip\\
$\pi$ & a permutation of $A$.\smallskip\\
$\Pi^A$ & the set of all possible permutations of $A$.\smallskip\\
$C^{\pi}_i$ & the coalition consisting of the agents that precede $i$ in $\pi$.\smallskip\\
$\mathbf{x}$ & a payoff vector, where $x_i\in\mathbf{x}$ is the payoff of agent $i\in A$.\smallskip\\
$\mc(i,C)$ & the marginal contribution of $i$ to $C$, i.e., $v(C\cup\{i\})-v(C)$.\smallskip\\
$\Sh_i$ & the Shapley value of $i$.\smallskip\\
$G$ & a graph.\smallskip\\
$G^C$ & the subgraph of $G$ that is induced by $C$.\smallskip\\
$N(G)$ & the set of nodes in $G$.\smallskip\\
$E(G) $ & the set of edges in $G$.\smallskip\\
$\adj(G',G)$ & the set of nodes in $G$ that are adjacent to $G'$, where $G'$ is an induced subgraph of $G$.\smallskip\\
$T$ & a tree.\smallskip\\
$T_i$ & the subset of $T$ rooted at $i$.\smallskip\\
$r(T)$ & the root of $T$.\smallskip\\
$\ch(i,T)$ & the set of children of $i$ in $T$.\smallskip\\
$\des(i,T)$ & the set of descendants of $i$ in $T$.\smallskip\\
$\anc(i,T)$ & the set of ancestors of $i$ in $T$.\smallskip\\
$\depth(i,T)$ & the depth of $i$ in $T$.\smallskip\\
$\Level_j(T)$ & the set of nodes in $N(T)$ whose depth is $j$.\smallskip\\
$\height(T)$ & the height of $T$.\smallskip\\
$T^C$ & the subgraph of $T$ that is induced by $C$.\smallskip\\
$T^C_{\trim}$ & the subgraph obtained by ``trimming'' $T^C$. In other words, it is the connected component in $T^C$ that contains $r(T)$. See Figure~\ref{fig:graphAndInducedSubgraph}\smallskip\\
$\Trimmed(T)$ & the set of coalitions that induce ``trimmed'' subgraphs. In other words, it contains every $C\subseteq N(T)$ such that $T^C = T^C_{\trim}$.\smallskip\\
$\SameTrim(C,T)$ & the set of coalitions $C'\subseteq N(C)$ such that by trimming $T^{C'}$ we obtain the same subgraph as the one obtained by trimming $T^C$.\smallskip
%
\end{longtable}


\end{document}